\documentclass[submission,copyright,creativecommons]{eptcs}

\pdfoutput=1
\usepackage{tikz}
\usetikzlibrary{automata,calc}
\usepackage[ruled]{algorithm2e}
\usepackage{footnote}
\usepackage{wrapfig}
\usepackage{amsfonts}
\usepackage{amsmath}
\usepackage{amsthm}
\usepackage{lineno}

\providecommand{\event}{GandALF 2023} 

\usepackage{iftex}

\ifpdf
  \usepackage{underscore}         
  \usepackage[T1]{fontenc}        
\else
  \usepackage{breakurl}           
\fi

\title{An Objective Improvement Approach \\ to Solving Discounted Payoff Games}
\author{Daniele Dell'Erba
\institute{University of Liverpool,\\ United Kingdom}
\email{\!\!\!\!daniele.dell-erba@liverpool.ac.uk\!\!\!\!}
\and
Arthur Dumas
\institute{ENS Rennes,\\ France}
\email{arthur.dumas@ens.rennes.fr}
\and
Sven Schewe
\institute{University of Liverpool,\\ United Kingdom}
\email{sven.schewe@liverpool.ac.uk}
}

\newcommand{\titlerunning}{An Objective Improvement Approach to Solving Discounted Payoff Games}
\newcommand{\authorrunning}{D. Dell'Erba, A. Dumas \& S. Schewe}

\hypersetup{
  bookmarksnumbered,
  pdftitle    = {\titlerunning},
  pdfauthor   = {\authorrunning},
  pdfsubject  = {EPTCS},               
}

\newcommand{\off}{\mathsf{offset}}
\newcommand\val{\mathsf{val}}
\newcommand\out{\mathsf{out}}
\newcommand{\denom}{\mathsf{denom}}
\newtheorem{obs}{Observation}
\newtheorem{thm}{Theorem}
\newtheorem{crl}{Corollary}
\newtheorem{lm}{Lemma}

\begin{document}
\maketitle

\begin{abstract}
While discounted payoff games and classic games that reduce to them, like parity and mean-payoff games, are symmetric, their solutions are not.
We have taken a fresh view on the constraints that optimal solutions need to satisfy, and devised a novel way to converge to them, which is entirely symmetric.
It also challenges the gospel that methods for solving payoff games are either based on strategy improvement or on value iteration.
\end{abstract}

\section{Introduction}
\label{sec:introduction}

We study turn-based zero sum games played between two
players on directed graphs.
The two players take turns to move a token along the
vertices of finite labelled graph with the goal to optimise their adversarial
objectives.

Various classes of graph games are characterised by the objective of the
players, for instance in \emph{parity games} the objective is to optimise the parity
of the dominating colour occurring infinitely often, while in \emph{discounted and mean-payoff games} the objective of the players is to minimise resp.\ maximise the discounted and limit-average sum of
the colours.  

Solving graph games is the central and most expensive step
in many model checking~\cite{Koz83,EJS93,Wil01,AHM01,AHK02,SF06},
satisfiability checking~\cite{Wil01,Koz83,Var98},
and synthesis~\cite{Pit06,SF06a} algorithms.
Progress in algorithms for solving graph games will therefore allow for the
development of more efficient model checkers and contribute to bringing synthesis
techniques to practice.

There is a hierarchy among the graph games mentioned earlier, with simple and well known reductions from parity games to mean payoff games, from mean-payoff games to discounted payoff games, and from discounted payoff games to simple stochastic games like the ones from \cite{ZP96}, while no reductions are known in the other direction. Therefore, one can solve instances of all these games by using an algorithm for stochastic games.
All of these games are in \textsf{UP} and \textsf{co-UP} \cite{Jur98}, while no tractable algorithm is known.

Most research has focused on parity games: as the most special class of games, algorithms have the option to use the special structure of their problems, and they are most directly linked to the synthesis and verification problems mentioned earlier.
Parity games have thus enjoyed a special status among graph games and the quest for efficient
algorithms \cite{EL86,EJ91,McN93,ZP96,BCJLM97,Zie98,Obd03,BDM18a,BDM16b}  
for solving them has been an active field of research during the last
decades, which has received further boost with the arrival of quasi-polynomial techniques \cite{CJKLS22,JL17,FJKSSW19,LB20,LPSW22,DS22}.

Interestingly, the one class of efficient techniques for solving parity games that does not (yet) have a quasi-polynomial approach is strategy improvement
algorithms~\cite{Lud95,Pur95,VJ00,BV07,Sch08a,Fea10a,STV15},
a class of algorithms closely related to the Simplex for linear programming, known to perform well in practice. 
Most of these algorithms reduce to mean~\cite{BV07,Sch08a,STV15,BDM20}
or discounted~\cite{Lud95,Pur95,FGO20,Koz21} payoff games.

With the exception of the case in which the fixed-point of discounted payoff games is explicitly computed~\cite{ZP96}, all these algorithms share a disappointing feature: they are inherently non-symmetric approaches for solving an inherently symmetric problem.
However, some of these approaches have a degree of symmetry. Recursive approaches treat even and odd colours symmetrically, one at a time, but they treat the two players very differently for a given colour.
Symmetric strategy improvement \cite{STV15} runs a strategy improvement algorithms for both players in parallel, using the intermediate results of each of them to inform the updates of the other, but at heart, these are still two intertwined strategy improvement algorithms that, individually, are not symmetric.
This is in due to the fact that applying strategy improvement itself symmetrically can lead to cycles~\cite{Con93}.

The key contribution of this paper is to devise a new class of algorithms to solve discounted payoff games, which is entirely symmetric.
Like strategy improvement algorithms, it seeks to find co-optimal strategies, and improves strategies while they are not optimal.
In order to do so, however, it does not distinguish between the strategies of the two players.
This seems surprising, as maximising and minimising appear to pull in opposing directions.

Similar to strategy improvement approaches, the new objective improvement approach turns the edges of a game into constraints (here called inequations), and minimises an objective function.
However, while strategy improvement algorithms take only the edges in the strategy of one player (and all edges of the other player) into account and then finds the optimal response by solving the resulting one player game, objective improvement always takes all edges into account.
The strategies under consideration then form a subset of the inequations, and the goal would be to make them sharp (i.e.\ as equations), which only works when both strategies are optimal.
When they are not, then there is some \emph{offset} for each of the inequations, and the objective is to reduce this offset in every improvement step.

This treats the strategies of both players completely symmetrically.
\medskip

\noindent\textbf{Organisation of the Paper.}
The rest of the paper is organised as follows.
After the preliminaries (Section \ref{sec:prelims}),
we start by outlining our method and use a simple game to explain it (Section \ref{sec:outline}).
We then formally introduce our objective improvement algorithm in Section \ref{sec:general}, keeping the question of how to choose a better strategies abstract.
Section \ref{sec:choose} then discusses how to find better strategies.
We finally wrap up with a discussion of our results in Section~\ref{sec:discuss}.

\section{Preliminaries}
\label{sec:prelims}

A \emph{discounted payoff game} (DPG) is a tuple $\mathcal{G} = (V_{\min}, V_{\max}, E, w, \lambda)$, where $V= V_{\min} \cup V_{\max}$ are the vertices of the game, partitioned into two disjoint sets $V_{\min}$ and $V_{\max}$, such that the pair $(V, E)$ is a finite directed
graph without sinks.
The vertices in $V_{\max}$ (\emph{resp}, $V_{\min}$) are controlled by Player Max or maximiser (\emph{resp}, Player Min or minimiser) and $E \subseteq V \times V$ is the edge relation. Every edge has a weight represented by the function $w : E \to \mathbb{R}$, and a \emph{discount factor} represented by the function $\lambda : E \to [0,1)$. When the discount factor is uniform, i.e.\ the same for every edge, it is represented by a constant value $\lambda \in [0,1)$.
For ease of notation, we write $w_e$ and $\lambda_e$ instead of $w(e)$ and $\lambda(e)$.
A \emph{play} on $\mathcal{G}$ from a vertex $v$ is an infinite path, which can be represented as a sequence of edges $\rho=e_0 e_1 e_2 \ldots$ such that, for every $i \in \mathbb{N}^*$, $e_i=(v_i,v'_{i})\in E$, and, for all $i \in \mathbb{N}, v_{i+1}=v'_{i}$ and $v_0=v$.
By $\rho_i$ we refer to the i-th edge of the play.
The \emph{outcome} of a discounted game $\mathcal{G} = (V_{\min}, V_{\max}, E, w, \lambda)$ for a play $\rho$ is 
$\out(\rho)=\sum_{i=0}^{\infty} w_{e_i} \prod_{j=0}^{i-1}\lambda_{e_j}$.
For games with a constant discount factor, this simplifies in $\out(\rho)=\sum_{i=0}^{\infty} w_{e_i} \lambda^i$.

A positional strategy for Max is a function $\sigma_{\max}:V_{\max}\to V$ that maps each Max vertex to a vertex according to the set of edges, i.e. $(v,\sigma_{\max}(v))\in E$.
Positional Min strategies are defined accordingly, and we call the set of positional Min and Max strategies $\Sigma_{\min}$ and $\Sigma_{\max}$, respectively.

A pair of strategies $\sigma_{\min}$ and $\sigma_{\max}$, one for each player, defines a unique run $\rho(v,\sigma_{\min},\sigma_{\max})$ from each vertex $v \in V$.
Discounted payoff games are positionally determined~\cite{ZP96}:
$$
\sup_{\sigma_{\max} \in \Sigma_{\max}} \inf_{\sigma_{\min} \in \Sigma_{\min}} \out(\rho(v,\sigma_{\min},\sigma_{\max}))
=
\inf_{\sigma_{\min} \in \Sigma_{\min}} \sup_{\sigma_{\max} \in \Sigma_{\max}} \out(\rho(v,\sigma_{\min},\sigma_{\max})) 
$$

holds for all $v\in V$, and neither the value, nor the optimal strategy for which it is taken, changes when we allow more powerful classes of strategies that allow for using memory and/or randomisation for one or both players.

The resulting \emph{value of $\mathcal G$}, denoted by $\val(\mathcal G): V \rightarrow \mathcal R$, is defined as
$$ \val(\mathcal G): v \mapsto  \sup_{\sigma_{\max} \in \Sigma_{\max}} \inf_{\sigma_{\min} \in \Sigma_{\min}} \out(\rho(v,\sigma_{\min},\sigma_{\max}))\ .$$

The solution to a discounted payoff game is a valuation $\val = \val(\mathcal G)$ of $\mathcal G$ for the vertices such that, for every edge $e = (v,v')$, it holds that\footnote{These are the constraints represented in $H$ in Section \ref{sec:general}.}
\begin{itemize}
    \item $\val(v) \leq w_e + \lambda_e \val(v')$ if $v$ is a minimiser vertex and
    \item $\val(v) \geq w_e + \lambda_e \val(v')$ if $v$ is a maximiser vertex.
\end{itemize}

A positional maximiser (resp.\ minimiser) strategy $\sigma$ is optimal if, and only if, $\val(v) = w_{(v,\sigma(v))} + \lambda_{(v,\sigma(v))} \val(\sigma(v))$ holds for all maximiser (resp.\ minimiser) positions.

Likewise, we define the value of a pair of strategies $\sigma_{\min}$ and $\sigma_{\max}$, denoted $\val(\sigma_{\min},\sigma_{\max}): V \rightarrow \mathcal R$, as
$ \val(\sigma_{\min},\sigma_{\max}): v \mapsto  \out(\rho(v,\sigma_{\min},\sigma_{\max}))$.

As we treat both players symmetrically in this paper, we define a \emph{pair of strategies} $\sigma : V \mapsto V$ whose restriction to $V_{\min}$ and $V_{\max}$ are a minimiser strategy $\sigma_{\min}$ and a maximiser strategy  $\sigma_{\max}$, respectively.
We then write $\rho(v,\sigma)$ instead of $\rho(v,\sigma_{\min},\sigma_{\max})$ and $\val(\sigma)$ instead of $\val(\sigma_{\min},\sigma_{\max})$.

If both of these strategies are optimal, we call $\sigma$ a joint \emph{co-optimal} strategy.
This is the case if, and only if, $\val(\mathcal G)= \val(\sigma)$ holds.

Note that we are interested in the \emph{value} of each vertex, not merely if the value is greater or equal than a given threshold value.

\section{Outline and Motivation Example}
\label{sec:outline}

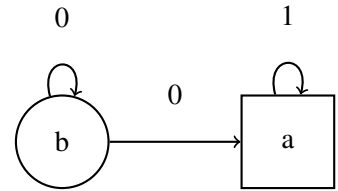
\begin{wrapfigure}{r}{0.4\textwidth} 
\vspace{-25pt}
  \begin{center}
\begin{tikzpicture}
    [node distance = 5em, bend angle = 22.5, inner sep = 0.3em, minimum size = 3.2em]
    \tikzset{every loop/.style = {max distance = 1.5em}}
    \node[shape=circle,draw,thick] (max) at (0,0) {b};
    \node[shape=rectangle,draw,thick] (min) at (3,0) {a};

    \draw[thick,->] (max) to node [above] {$0$} (min);
    \draw [thick,->] (max) edge[loop above]node{$0$} (max);
    \draw [thick,->] (min) edge[loop above]node{$1$} (min);
    \end{tikzpicture}
\caption{A discounted Payoff Game. Maximiser vertices are marked by a circle, minimizer ones by a square.}
\label{fig:example}
  \end{center}
\end{wrapfigure}
We start with considering the simple discounted payoff game of Figure~\ref{fig:example}, assuming that it has some uniform discount factor $\lambda \in [0,1)$.
In this game, the minimiser (who owns the right vertex, $a$, marked by a square), has only one option: she always has to use the self-loop, which earns her an immediate reward of $1$.
The overall reward the minimiser reaps for a run that starts in her vertex is therefore 
 $1 + \lambda + \lambda^2 + \ldots = \frac{1}{1-\lambda}$.
 
The maximiser (who owns the left vertex, $b$, marked by a circle) can choose to either use the self-loop, or to move to the minimiser vertex (marked by a square), both yielding no immediate reward.

If the maximiser decides to stay forever in his vertex (using the self-loop), his overall reward in the play that starts at (and, due to his choice, always stays in) his vertex, is $0$.
If he decides to move on to the minimiser vertex the $n^{th}$-time, 
then the reward is $\frac{\lambda^n}{1-\lambda}$.

The optimal decision of the maximiser is therefore to move on the first time, which yields him the maximal reward of $\frac{\lambda}{1-\lambda}$.
Every vertex $v$ has some outgoing edge(s) $e = (v,v')$ where $\val(v) = w_e + \lambda_e \val(v')$ holds~\cite{ZP96};
these edges correspond to the optimal decisions for the respective player.

For our running example game of Figure~\ref{fig:example} with a fixed discount factor $\lambda \in [0,1)$, the inequations are
\begin{enumerate}
    \item $\val(a) \leq 1 + \lambda \val(a)\quad $ for the self-loop of the minimiser vertex;
    \item $\val(b) \geq \lambda \val(b)\qquad\ \ $ for the self-loop of the maximiser vertex; and
    \item $\val(b) \geq \lambda \val(a)\qquad\ \ $ for the transition from the maximiser to the minimiser vertex.
\end{enumerate}
The unique valuation that satisfies these inequations and produces a sharp inequation (i.e. satisfied as equation) for some outgoing edge of each vertex assigns $\val(a)=\frac{1}{1-\lambda}$ and $\val(b)=\frac{\lambda}{1-\lambda}$.
This valuation also defines the optimal strategies of the players (to stay for the minimiser, and to move on for the maximiser).

Solving a discounted payoff game means finding this valuation and/or these strategies.

We discuss a symmetric approach to find this unique valuation.
Our approach adjusts linear programming in a natural way that treats both players symmetrically:
we maintain the set of inequations for the complete time, while approximating the goal of ``one equality per vertex'' by the objective function.
To do that, we initially fix an \emph{arbitrary} outgoing edge for every vertex (a strategy), and minimise the sum of the distances between the left and right side of the inequations defined by these edges, which we call the \emph{offset} of this edge. This means, for an edge $e = (v,v')$, to minimise the difference of $\val(v)$ (left side of the inequation) and $w_e + \lambda_e \val(v')$ (right side).

To make this clear, we consider again the example of Figure~\ref{fig:example} and use both self loops as the strategies for the players fixed at the beginning in our running example.
The offset for the selected outgoing edge of the minimiser vertex $a$ is equal to $1 -(1-\lambda)\val(a)$, while the offset for the selected outgoing edge of the maximiser vertex $b$ is equal to $(1-\lambda)\val(b)$.
The resulting overall objective consists, therefore, in minimising the value $1-(1-\lambda)\val(a) + (1-\lambda)\val(b)$.

This term is always non-negative, since it corresponds to the sum of the edges contributions that are all non-negative.
Moreover, when only optimal strategies are selected to form this objective function, the value $0$ can be taken, and where it is taken, it defines the correct valuation.

As the maximiser's choice to take the self-loop is not optimal, the resulting objective function the strategies define, that is $1-(1-\lambda)\val(a) + (1-\lambda)\val(b)$, cannot reach $0$.
But let us take a look at what an optimal solution w.r.t.\ this objective function looks like.

Optimal solutions can be taken from the corners of the polytope defined by the inequations.
In this case, the optimal solution (w.r.t.\ this initial objective function) is defined by making inequations (1) and (3) sharp: this provides the values $\val(a)=\frac{1}{1-\lambda}$ and $\val(b)=\frac{\lambda}{1-\lambda}$; the objective function takes the value $\lambda$ at this point.

For comparison, in the other corner of the polytope, defined by making inequations (2) and (3) sharp, we obtain the values $\val(a)=\val(b)=0$; the objective function takes the value $1$ at this point.
Finally, if we consider the last combination, making (1) and (2) sharp provides the values $\val(a)=\frac{1}{1-\lambda}$ and $\val(b)=0$, so that inequation (3) is not satisfied; this is therefore not a corner of the polytope.

Thus, in this toy example, while selecting the wrong edge cannot result in the objective function taking the value $0$, we still found the optimal solution.
In general, we might need to update the objective function.
To update the objective function, we change the outgoing edges of some (or all) vertices, such that the overall value of the objective function goes down.
Note that this can be done not only when the linear program returns an optimal solution, but also during its computation.
For example, when using a simplex method, updating the objective function can be used as an alternative pivoting rule at any point of the traversal of the polytope.

Unfortunately, the case in which the valuation returned as solution is computed using an objective function based on non-optimal strategies, is not the general case.
The simplest way of seeing this is to use different discount factors for the game of Figure~\ref{fig:example}%
\footnote{Note that we can also replace the transitions with a smaller discount factor by multiple transitions with a larger discount factor. This would allow for keeping the discount factor uniform, but needlessly complicate the discussion and inflate the size of the example.}, let's say $\frac{1}{3}$ for the self-loop of the maximiser vertex and $\frac{2}{3}$ for the other two transitions, so that the three equations are: (1) $\val(a) \leq 1 + \frac{2}{3} \val(a)$, (2) $\val(b) \geq \frac{2}{3} \val(b)$, and (3) $\val(b) \geq \frac{1}{3} \val(a)$.
Making the adjusted inequations (2) and (3) sharp still results in the values $\val(a)=\val(b)=0$, and the objective function still takes the value of $1$. While making inequations (1) and (3) sharp provides the values $\val(a)=3$ and $\val(b)=2$; the objective function takes the value $\frac{4}{3}$ at this point.
Finally, if we consider the last combination, making (1) and (2) sharp still conflicts with inequation (3).

Thus, $\val(a)=\val(b)= 0$ would be the optimal solution for the given objective function, which is not the valuation of the game. 
We will then update the candidate strategies so that the sum of the offsets goes down.

\subsection{Comparison to strategy improvement}
The closest relative to our new approach are strategy improvement algorithms.
Classic strategy improvement approaches solve this problem of finding the valuation of a game (and usually also co-optimal strategies) by (1) fixing a strategy for one of the players (we assume w.l.o.g.\ that this is the maximiser), (2) finding a valuation function for the one player game that results from fixing this strategy (often together with an optimal counter strategy for their opponent), and (3) updating the strategy of the maximiser by applying local improvements.
This is repeated until no local improvements are available, which entails that the constraint system is satisfied.

For Step (2) of this approach, we can use linear programming, which does invite a comparison to our technique.
The linear program for solving Step (2) would not use all inequations: it would, instead, replace the inequations defined by the currently selected edges of the maximiser by equations, while dropping the inequations defined by the other maximiser transitions.
The objective function would then be to minimise the values of all vertices while still complying with the remaining (in)equations.

Thus, while in our novel symmetric approach the constraints remain while the objective is updated, in strategy improvement the objective remains, while the constraints are updated.

Moreover, the players and their strategies are treated quite differently in strategy improvement algorithms: while the candidate strategy of the maximiser results in making the inequations of the selected edges sharp (and dropping all other inequations of maximiser edges), the optimal counter strategy is found by minimising the objective. This is again in contrast to our novel symmetric approach, which treats both players equally.

A small further difference is in the valuations that can be taken: the valuations that strategy improvement algorithms can take are the valuations of strategies, while the valuations our objective improvement algorithm can take on the way are the corners of the polytope defined by the inequations.
Except for the only intersection point between the two (the valuation of the game), these corners of the polytope do not relate to the value of strategies. 
Table \ref{tab:comparison} summarises these observations.

\renewcommand{\arraystretch}{0.9}
\begin{table}[t]
    \centering
    \begin{tabular}{c||c|c}
         & \textbf{Objective Improvement} & \textbf{Strategy Improvement} \\[0.5em]
        \hline
        \textbf{players}  & symmetric treatment & asymmetric treatment \\[0.5em]
        \textbf{constraints} & remain the same: & change: \\
         & one inequation per edge & one inequation for each edge \\
         & & defined by the current strategy for \\
         & & the strategy player, one inequation \\
         & & for every edge of their opponent \\[0.5em]
        \textbf{objective} & minimise errors for selected edges & maximise values \\[0.5em]
        \textbf{update} & objective: & strategy: \\
         & one edge for each vertex & one edge for each vertex \\
         & & of the strategy player  \\[0.5em]
        \textbf{valuations} & corners of polytope & defined by strategies
         
    \end{tabular}
    \vspace{1em}
    \caption{A comparison of the
    novel objective improvement with classic strategy improvement.}
    \label{tab:comparison}
\end{table}

\section{General Objective Improvement}
\label{sec:general}
In this section, we present the approach outlined in the previous section more formally, 
while keeping the most complex step -- updating the candidate strategy to one which is \emph{better} in that it defines an optimisation function that can take a smaller value -- abstract. (We will turn back to the question of how to find better strategies in Section \ref{sec:choose}.)
This allows for discussing the principal properties more clearly.

A general outline of our \emph{objective improvement} approach is based on
this algorithm:

\begin{wrapfigure}{r}{0.48\textwidth}
\begin{algorithm}[H]
  \caption{\label{alg:algname} Objective Improvement}
  \SetKwInOut{Input}{input} \SetKwInOut{Output}{output}
  
  \Input{A discounted payoff game $\mathcal{G} = (V_{\min}, V_{\max}, E, w, \lambda)$}
  \Output{The valuation $\val$ of $\mathcal G$}
  \SetInd{0.25em}{0.5em}
    {
    \nl $H \leftarrow{} \mathsf{Inequations}(\mathcal{G})$\\
    \nl $\sigma \leftarrow{} \mathsf{ChooseInitialStrategies}(\mathcal{G})$\\
    \nl \While{\sf{true}}{
        \nl $f_\sigma \leftarrow{} \mathsf{ObjectiveFunction}(\mathcal{G}, \sigma)$ \\
        \nl $\val \leftarrow{} \mathsf{LinearProgramming}(H, f_\sigma)$\\
        \nl \If{ $f_\sigma(\val) = 0$}{\Return $\val$}
            \nl $\sigma \leftarrow{} \mathsf{ChooseBetterStrategies}(\mathcal{G},\sigma)$
        }
    }
\end{algorithm}
  \vspace{-11pt}
\end{wrapfigure}

Before describing the procedures called by the algorithm, we first outline the principle.

When running on a discounted payoff game $\mathcal{G}= (V_{\min}, V_{\max}, E, w, \lambda)$, the algorithm uses a set of inequations defined by the edges of the game and the owner of the source of each edge.
This set of inequations denoted by $H$ contains one inequation for each edge and (different to strategy improvement approaches whose set of inequations is a subset of $H$) $H$ never changes.

The inequations from $H$ are computed by a function called $\mathsf{Inequations}$ that, given the discounted game $\mathcal{G}$, returns the set made of one inequation per edge $e=(v,v') \in E$, defined as follows:
\[ I_e =  \begin{cases}
        \val(v) \geq w_{
        e} + \lambda_e \val(v') &\text{if }v \in V_{\max} \\
        \val(v) \leq w_{
        e} + \lambda_e \val(v') &\text{otherwise .}
    \end{cases} 
\]

The set $H = \{I_e \mid e \in E\}$ is
defined as the set of all inequations for the edges of the game.

The algorithm also handles strategies for both players, treated as a single strategy $\sigma$.
They are initialised (for example randomly) by the function $\mathsf{ChooseInitialStrategies}$.

This joint strategy is used to define an objective function $f_{\sigma}$ by calling function $\mathsf{ObjectiveFunction}$, whose value on an evaluation $\val$ is: $f_{\sigma}(\val)=\sum_{v \in V} f_\sigma(\val, v)$ with the following objective function components: \\
\[
  f_\sigma(\val, v) = \off(\val,(v,\sigma(v)))\ ,
\]
where the offset of an edge $(v,v')$ for a valuation is defined as follows:\\
\[
  \off(\val, (v,v')) =
  \begin{cases}
    \val(v) - (w_{(v,v')} + \lambda_{(v,v')} \val(v')) &\text{if } v \in V_{\max} \\
    (w_{(v,v')} + \lambda_{(v,v')} \val(v')) - \val(v) &\text{otherwise}
  \end{cases}
\]

This objective function $f_\sigma$ is given to a linear programming algorithm, alongside with the inequations set $H$.
We underline that, due to the inequation to $I_{(v,v')}$, the value of $\off(\val, (v,v'))$ is non-negative for all $(v,v')\in E$ in any valuation $\val$ (optimal or not) that satisfies the system of inequations $H$.
We put a further restriction on $\val$ in that we require it to be the solution to a \emph{basis} $\mathbf{b}$ in $H$.
Such a basis consists of $|V|$ inequations that are satisfied sharply (i.e.\ as equations), such that these $|V|$ equations uniquely define the values of all vertices.
We refer to this valuation as the evaluation of $\mathbf{b}$, denoted $\val(\mathbf{b})$.

The call $\mathsf{LinearProgramming}(H, f_\sigma)$ to some linear programming algorithm then returns a valuation $\val$ of the vertices that minimise $f_\sigma$ while satisfying $H$, and for convenience require this valuation to also be $\val(\mathbf{b})$ for some base $\mathbf{b}$ of $H$. (Note that the simplex algorithm, for example, only uses valuations of this form in every step.)
We call this valuation a \emph{valuation associated to $\sigma$}.

\begin{obs}
\label{obs:nonneg}
At Line 6 of Algorithm \ref{alg:algname}, the value of $f_\sigma(\val)$ is non-negative.
\end{obs}

We say that a valuation $\val$ \emph{defines} strategies of both players if, for every vertex $v\in V$, the inequation of (at least) one of the outgoing edges of $v$ is sharp.
These are the strategies defined by using, for every vertex $v \in V$, an outgoing edge for which the inequation is sharp.
Note that there can be more than one of these inequations for some of the vertices.

\begin{obs}
\label{obs:0strategy}
If, for a solution $\val$ of $H$, $f_\sigma(\val) = 0$ holds, then, for every vertex $v \in V$, the inequation $I_{(v,\sigma(v))}$ for the edge $(v,\sigma(v))$ is sharp, and $\val$ therefore defines strategies for both players, those defined by $\sigma$, for example.
\end{obs}

We can use, alternatively, $f_\sigma(\val) = 0$ as a termination condition, as shown in Algorithm \ref{alg:algname}, since in this case $\sigma$ must define co-optimal strategies.

\begin{thm}
\label{thm:main}If $\sigma$ describes co-optimal strategies, then $f_\sigma(\val) = 0$ holds at Line 6 of Algorithm \ref{alg:algname}.
If $\val$ defines strategies for both players joint in $\sigma$ at Line 6 of Algorithm \ref{alg:algname}, then $\sigma$ is co-optimal and $\val$ is the valuation of $\mathcal G$.
\end{thm}

\begin{proof}
The valuation $\val=\val(\mathcal G)$ of the game is the unique solution of $H$ for which, for all vertices $v$, the inequation to (at least) one of the outgoing edges of $v$ is sharp, and the edges for which they are sharp describe co-optimal strategies.
The valuation of the game is thus the only valuation that \emph{defines} strategies for both players, which shows the second claim.

Moreover, if $\sigma$ describes co-optimal strategies, then $f_\sigma (\val)=0$ holds for $\val=\val(\mathcal G)$ (and for this valuation only), which establishes the first claim.
\end{proof}

The theorem above ensures that, in case the condition at Line 6 holds, the algorithm terminates and provides the value of the game that then allows us to infer optimal strategies of both players.
Otherwise we have to improve the objective function and make another iteration of the while loop.
At Line 7, $\mathsf{ChooseBetterStrategies}$ can be any procedure that, for $f_\sigma(\val) \neq 0$, provides a pair of strategy $\sigma'$
\emph{better} than $\sigma$ as defined in the following subsection.

\paragraph{Better strategies}

A strategy $\sigma'$ for both players is \emph{better} than a strategy $\sigma$ if, and only if, the minimal value of the objective function $f_{\sigma'}$ (computed by $\mathsf{LinearProgramming}(H, f_{\sigma'})$) is strictly lower than the minimal value of the objective function for $f_\sigma$ (computed by $\mathsf{LinearProgramming}(H, f_{\sigma})$). Formally, $f_{\sigma'}(\val') < f_\sigma(\val)$.

While we discuss how to implement this key function in the next section, we observe here that the algorithm terminates with a correct result with any implementation that chooses a better objective function in each round: correctness is due to it only terminating when $\val$ \emph{defines} strategies for both players, which implies (cf.\ Theorem \ref{thm:main}) that $\val$ is the valuation of $\mathcal G$ ($\val=\val(\mathcal G)$) and all strategies defined by $\val$ are co-optimal.
Termination is obtained by a finite number of positional strategies: by Observation \ref{obs:nonneg}, the value of the objective function of all of them is non-negative, while the objective function of an optimal solution to co-optimal strategies is $0$ (cf.\ Theorem \ref{thm:main}), which meets the termination condition of Line 6 (cf.\ Observation \ref{obs:0strategy}).

\begin{crl}
Algorithm \ref{alg:algname} always terminates with the correct value.
\end{crl}
\section{Choosing Better Strategies}
\label{sec:choose}
In this section, we will discuss sufficient criteria for efficiently finding a procedure that implements $\mathsf{ChooseBetterStrategies}$.
For this, we make four observations described in the next subsections:

\begin{enumerate}
    \item All local improvements can be applied.
        A strategy $\sigma'$ is a local improvement to a strategy $\sigma$ if $f_{\sigma'}(\val) < f_\sigma(\val)$ holds for the current valuation $\val$ (Section~\ref{ssec:local}).
    
    \item If the current valuation $\val$ does not \emph{define} a pair of strategies $\sigma$ for both players and has no local improvements, then a better strategy $\sigma'$ can be found by applying only switches from and to edges that already have offset $0$ (Section \ref{ssec:noLocal}).
    
    \item The improvement mentioned in the previous point can be found for special games (the sharp and improving games defined in Section \ref{ssec:efficient}) by trying a single edge switch.
    
    \item Games can almost surely be made sharp and improving by adding random noise that retains optimal strategies (Section \ref{ssec:sharpeing}).
\end{enumerate}

Together, these four points provide efficient means for finding increasingly better strategies, and thus to find the co-optimal strategies and the valuation of the discounted payoff game.

As a small side observation, when using a simplex based technique to implement $\mathsf{LinearProgramming}$ at Line 5 of Algorithm \ref{alg:algname}, then the pivoting of the objective function from point (1.) and the pivoting of the base can be mixed (this will be discussed in Section \ref{ssec:mixing}).

\subsection{Local Improvements}
\label{ssec:local}

The simplest (and most common) case of creating better strategies $\sigma'$ from a valuation for the objective $f_\sigma$ for a strategy $\sigma$ is to consider \emph{local improvements}.
Much like local improvements in strategy iteration approaches, local improvements consider, for each vertex $v$, a successor $v' \neq \sigma(v)$, such that $\off(\val,(v,v')) < \off(\val,(v,\sigma(v))$ for the current valuation $\val$, which is optimal for the objective function $f_\sigma$.

To be more general, our approach does not necessarily requires to select only local improvements, but it can work with global improvements, though we cannot see any practical use of choosing differently.
For instance, if we treat the function as a global improvement approach, we can update the value for a vertex $v$ such that it increases by 1 and update the value of another vertex $v'$ such that it decreases by 2. The overall value of the function will decrease, even if locally some components increased their value. Interestingly, this cannot be done with a strategy improvement approach, as it requires to always locally improve the value of each vertex when updating.

\begin{lm}
\label{lem:better}
If $\val$ is an optimal valuation for the linear programming problem at Line 5 of Algorithm \ref{alg:algname} and $f_{\sigma'}(\val) < f_\sigma(\val)$, then $\sigma'$ is better than $\sigma$.
\end{lm}
\begin{proof}
The valuation $\val$ is, being an optimal solution for the objective $f_\sigma$, a solution to the system of inequations $H$. For a solution $\val'$ which is optimal for $f_{\sigma'}$, we thus have $f_{\sigma'}(\val') \leq f_{\sigma'}(\val) < f_\sigma(\val)$, which implies that $\sigma'$ is better than $\sigma$ accordingly to notion of \emph{better} strategy provided at the end of Section 4.
\end{proof}
\subsection{No Local Improvements}
\label{ssec:noLocal}

The absence of local improvements means that, for all vertices $v \in V$ and all outgoing edges $(v,v') \in E$, $\off(\val,(v,v')) \geq \off(\val,(v,\sigma(v)))$.

We define for a valuation $\val$
optimal for a $f_\sigma$ (like the $\val$ produced in line 5 of Algorithm~\ref{alg:algname}):
\begin{itemize}
    \item $S_\val^\sigma = \{ (v,v') \in E \mid \off(\val,(v,v')) = \off(\val,(v,\sigma(v)))\}$ as the set of \emph{stale} edges;
    naturally, every vertex has at least one outgoing stale edge: the one defined by $\sigma$;
    
    \item $E_\val = \{ (v,v') \in E \mid \off(\val,(v,v')) =  0\}$ as the set of edges, for which the inequation for $\val$ is sharp;
    in particular, all edges in the base of $H$ that defines $\val$ are sharp (and stale); and
    
    \item $E_\val^\sigma$ as any set of edges between $E_\val$ and $S_\val^\sigma$ (i.e.\ $E_\val \subseteq E_\val^\sigma \subseteq S_\val^\sigma$) such that $E_\val^\sigma$ contains an outgoing edge for every vertex $v \in V$;
    we are interested to deal with sets that retain the game property that every vertex has a successor, we can do that by adding (non sharp) stale edges to $E_\val$.
\end{itemize}  

Note that $S_\val^\sigma$ is such a set, hence, an adequate set is easy to identify.
We might, however, be interested in keeping the set small, and choosing the edges defined by $E_\val$ plus one outgoing edge for every vertex $v$ that does not have an outgoing edge in $E_\val$.
The most natural solutions is to choose the edge $(v,\sigma(v)) \in E_\val^\sigma$ defined by $\sigma$ for each such vertex $v$.

\begin{obs}
\label{obs:stalegame}
If $\mathcal{G}= (V_{\min}, V_{\max}, E, w, \lambda)$ is a DPG and $\sigma$ a strategy for both players such that $\val$ is an optimal solution for the objective $f_\sigma$ to the system of inequations $H$,
then $\mathcal G'= (V_{\min}, V_{\max}, E_\val^\sigma, w, \lambda)$ is also a DPG.
\end{obs}

This simply holds because every vertex $v \in V$ retains at least one outgoing transition.

\begin{lm}
\label{lem:stale}
Let $\mathcal{G}= (V_{\min}, V_{\max}, E, w, \lambda)$ be a DPG, $\sigma$ a strategy for both players, $\val$ an optimal solution returned at Line 5 of Algorithm \ref{alg:algname} for $f_\sigma$, and let there be no local improvements of $\sigma$ for $\val$.
If $\val$ does not define strategies of both players, then there is a better strategy $\sigma'$ such that, for all $v \in V$, $(v,\sigma'(v)) \in E_\val^\sigma$.
\end{lm}
\begin{proof}
By Observation \ref{obs:stalegame},  $\mathcal G'=  (V_{\min}, V_{\max}, E_\val^\sigma, w, \lambda)$ is a DPG. Let $\val'$ be the valuation of $\mathcal G'$, and $\sigma'$ be the strategies for the two players defined by it.

If $\val'$ is also a valuation of $\mathcal G$, then we are done. However, this need not be the case, as the system of inequations $H'$ for $\mathcal G'$ is smaller than the set of inequations $H$ for $\mathcal G$, so $\val'$ might violate some of the inequations that are in $H$, but not in $H'$.
Given that $\val'$ is a valuation for $\mathcal G'$, it satisfies all inequations in $H'$.
Moreover, since $\val$ also satisfies all inequations of $H'$, it follows that the same inequations hold for every convex combination of $\val$ and $\val'$.

We now note that the inequations of $H$ that are not in $H'$ are not sharp for $\val$. Thus, there is an $\varepsilon \in (0,1]$ such that the convex combination $\val_\varepsilon = \varepsilon \cdot \val' + (1-\varepsilon) \val$ is a solution to those inequations.

We now have $f_{\sigma'}(\val) = f_\sigma(\val) > 0$ and $f_{\sigma'}(\val')=0$.
For an optimal solution $\val''$ of $H$ for the objective $f_{\sigma'}$, this provides $f_{\sigma'}(\val'') \leq f_{\sigma'}(\val_\varepsilon) < f_\sigma(\val)$.

Therefore $\sigma'$ is better than $\sigma$.
\end{proof}

When using the most natural choice, $E_\val^\sigma = E_\val \cup \{(v,\sigma(v))\mid v \in V\}$, this translates in keeping all transitions, for which the offset is \emph{not} $0$, while changing some of those, for which the offset already is $0$. This is a slightly surprising choice, since to progress one has to improve on the transitions whose offset is positive, and ignore those with offset 0.

\subsection{Games with Efficient Objective Improvement}
\label{ssec:efficient}

In this subsection, we consider sufficient conditions for finding better strategies efficiently.
Note that we only have to consider cases where the termination condition (Line 6 of Algorithm~\ref{alg:algname}) is not met.

The simplest condition for efficiently finding better strategies is the existence of local improvements.
(In particular, it is easy to find, for a given valuation $\val$, strategies $\sigma'$ for both players such that $f_{\sigma'}(\val) \leq f_{\sigma''}(\val)$ holds for all strategies $\sigma''$).
When there are local improvements, we can obtain a better strategy simply by applying them.

This leaves the case in which there are no local improvements, but where $\val$ also does not \emph{define} strategies for the two players.
We have seen that we can obtain a better strategy by only swapping edges, for which the inequations are sharp (Lemma \ref{lem:stale}).

We will now describe two conditions that, when both met, will allow us to efficiently find better strategies:
that games are \emph{sharp} and \emph{improving}.
\medskip

\noindent\textbf{Sharp games.}
To do this efficiently, it helps if there are always $|V|$ inequations that are sharp: there must be at least $|V|$ of them for a solution returned by the simplex method, as it is the solution defined by making $|V|$ inequations sharp (they are called the base), and requiring that there are exactly $|V|$ many of them means that the valuation we obtain defines a base.
We call such a set of inequations $H$, and games that define them \emph{sharp DPGs}.
\medskip

\noindent\textbf{Improving games.}
The second condition, which will allow us to identify better strategies efficiently, is to assume that, for every strategy $\sigma$ for both players, if a valuation $\val$ defined by a base is not optimal for $f_\sigma$ under the constraints $H$, then there is a single base change that improves it.
We call such sharp DPGs \emph{improving}.

We call a valuation $\val'$ whose base can be obtained from that of $\val$ by a single change to the base of $\val$ a neighbouring valuation to $\val$. 
We will show that, for improving games, we can sharpen the result of Lemma \ref{lem:stale} so that the better strategy $\sigma'$ also guarantees $f_{\sigma'}(\val')<f_\sigma(\val)$ for some neighbouring valuation $\val'$ to $\val$. 

This allows us to consider $O(|E|)$ base changes and, where they define a valuation, to seek optimal strategies for a given valuation.
Finding an optimal strategy for a given valuation is straightforward.

\begin{thm}
\label{theo:improve}
Let $\mathcal{G}= (V_{\min}, V_{\max}, E, w, \lambda)$ be an improving DPG, $\sigma$ a strategy for both players, $\val$ an optimal solution returned at Line 5 of Algorithm \ref{alg:algname} for $f_\sigma$, and let there be no local improvements of $\sigma$ for $\val$.
Then there is (at least) one neighbouring valuation $\val''$ to $\val$ such that there is a better strategy $\sigma'$ that satisfies $f_{\sigma'}(\val'')<f_\sigma(\val)$.

Such a strategy $\sigma'$ is better than $\sigma$, and it can be selected in a way that $(v,\sigma'(v)) \in E_\val^\sigma$ holds for all $v\in V$ for a given set $E_\val^\sigma$.
\end{thm}
\begin{proof}
We apply Lemma~\ref{lem:stale} for $E_\val^\sigma = E_\val \cup \{(v,\sigma(v)\mid v \in V\}$ and use the resulting better strategy $\sigma'$ for this set $E^\sigma_\val$.
Let $\val'$ be the optimal solution for $f_{\sigma'}$ that satisfies the constraints $H$ defined by $\mathcal G$.
Note that since $\sigma'$ is better than $\sigma$ by Lemma \ref{lem:stale}, we have that  $f_{\sigma'}(\val')<f_\sigma(\val)$ and $f_\sigma(\val) \leq f_\sigma(\val')$.

We now consider an arbitrary sequence of evaluations $\val = \val_0, \val_1,\ldots,\val_n = \val'$ along the edges of the simplex from $\val$ to $\val'$, such that the value of the new objective function $f_{\sigma'}$ only decreases.
Note that such a path must exist, as the simplex algorithm would pick it.

The sharpness of $\mathcal G$ implies that $\val_1 \neq \val_0$, and considering that $\mathcal G$ is improving provides $f_{\sigma'}(\val_1)<f_{\sigma'}(\val_0)$.

Thus, when only applying a single base change, we move to a fresh value, $\val_1$, such that $f_{\sigma'}(\val_1) < f_\sigma(\val)$ for some $\sigma'$.

Note that $\sigma'$ was supplied by Lemma \ref{lem:stale}, so that $(v,\sigma'(v)) \in E_\val^\sigma$ holds.
\end{proof}

While we can restrict the selection of $\sigma'$ to those that comply with the restriction $(v,\sigma'(v)) \in E_\val^\sigma$, there is no particular reason for doing so; as soon as we have a neighbouring valuation $\val'$, we can identify a pair of strategies $\sigma'$ for which $f_{\sigma'}(\val')$ is minimal, and select $\sigma'$ if $f_{\sigma'}(\val')< f_\sigma(\val)$ holds.

\subsection{Making Games Sharp and Improving}
\label{ssec:sharpeing}

Naturally, not every game is improving, or even sharp. In this subsection, we first discuss how to almost surely make games sharp by adding sufficiently small random noise to the edge weights, and then discuss how to treat games that are not improving by assigning randomly chosen factors, with which the offsets of edges are weighted.
Note that these are `global' adjustments of the game that only need to be applied once, as it is the game that becomes sharp and improving, respectively.

Starting with the small noise to add on the edge weights, we first create notation for expressing how much we can change edge weights, such that joint co-optimal strategies of the resulting game are joint co-optimal strategies in the original game.
To this end, we define the \emph{gap} of a game.
\medskip

\noindent\textbf{Gap of a game.}
For a DPG  $\mathcal{G}= (V_{\min}, V_{\max}, E, w, \lambda)$, we call $\lambda^*=\max\{\lambda_e\mid e \in E\}$ its contraction. For a joint strategy $\sigma$ that is \emph{not} co-optimal, we call $\gamma_\sigma = - \min \{\off(\val(\sigma),e) \mid e \in E\}$; note that $\gamma_\sigma >0$ holds.
We call the minimal%
\footnote{This skips over the case where all strategies are co-optimal, but that case is trivial to check and such games are trivial to solve, so that we ignore this case in this subsection.}
such value $\gamma$ the \emph{gap of $\mathcal G$}.

Note that $\val(\sigma)$ is the valuation of the joint strategy $\sigma$, not the outcome of the optimisation problem. This is because we use the gap of the game to argue that a non-co-optimal strategy remains non-co-optimal after a small distortion of the edge weights, so that the value of the joint strategy itself is useful. (It is much easier to access than the result of optimisation.) This also implies that the offsets can be negative. 

We now use the gap of a game $\gamma$ to define the magnitude of a change to all weights, such that all strategies that used to have a gap still have one.

\begin{lm}\label{lem9}
Let $\mathcal G = (V_{\min}, V_{\max}, E, w, \lambda)$ be a DPG with
contraction $\lambda^*$ and gap $\gamma$, and let \\ $\mathcal G' = (V_{\min}, V_{\max}, E, w', \lambda)$ differ from $\mathcal G$ only in the edge weights such that, for all $e \in E$, $|w_e - w_e'| \leq \frac{1-\lambda^*}{3}\gamma$ holds.
Then a joint co-optimal strategy from $\mathcal G'$ is also co-optimal for $\mathcal G$.
\end{lm}
\begin{proof}
The small weight disturbance, $|w_e - w_e'| \leq \frac{1-\lambda^*}{3}\gamma$ for all $e \in E$, immediately provides a small difference in the valuation: for all joint strategies $\sigma$, we have for $\val = \val(\sigma)$ on $\mathcal G$, and $\val' = \val(\sigma)$ on $\mathcal G'$, that $|\val(v) - \val'(v)| \leq \frac{1}{1-\lambda^*} \frac{1-\lambda^*}{3}\gamma = \frac{\gamma}{3}$, using the definition of $\val(\sigma)$ and triangulation.

More precisely, as $\sigma$ defines a run $\rho = e_0 e_1 e_2 \ldots$, and we have $\val(v) = \out(\rho)=\sum_{i=0}^{\infty} w_{e_i} \prod_{j=0}^{i-1}\lambda_{e_j}$
and $\val'(v)=\sum_{i=0}^{\infty} w_{e_i}' \prod_{j=0}^{i-1}\lambda_{e_j}$.
This provides $$|\val(v) - \val'(v)| \leq \sum_{i=0}^{\infty} |w_{e_i}-w_{e_i}'| \prod_{j=0}^{i-1}\lambda_{e_j} < \frac{1-\lambda^*}{3}\gamma \sum_{i=0}^{\infty}\prod_{j=0}^{i-1}\lambda_{e_j} \leq \frac{1-\lambda^*}{3}\gamma \sum_{i=0}^{\infty} (\lambda^*)^i = \frac{\gamma}{3}. $$

If $\sigma$ is not co-optimal for $\mathcal G$, we have an edge $e$ with $-\off(\val,e)=\gamma_\sigma \geq \gamma$.
Triangulation provides $$|\off(\val',e)-\off(\val,e)| < \frac{1+\lambda^*}{3}\gamma$$ and (using $\off'$ to indicate the use of $w_e'$ for $\mathcal G'$ instead of $w_e$ for $\mathcal G$), 
$$|\off'(\val',e)-\off(\val,e)| \leq \frac{2+\lambda^*}{3}\gamma< \gamma \leq \gamma_\sigma\ ,$$
which, together with the fact that $-\off(\val,e) \geq \gamma$, provides $\off'(\val',e)<0$.

Thus, $\sigma$ is not co-optimal for $\mathcal G'$.
\end{proof}

\begin{lm}\label{l10}
Given a DPGs $\mathcal{G}= (V_{\min}, V_{\max}, E, w, \lambda)$, the DPG $\mathcal G'= (V_{\min}, V_{\max}, E, w', \lambda)$ resulting from $\mathcal G$ by adding independently uniformly at random drawn values from an interval $(-\varepsilon,\varepsilon)$ to every edge weight, will almost surely result in a sharp game.
\end{lm}
\begin{proof}
There are only finitely many bases, and it suffices to compare two arbitrary but fixed ones, say $b_1$ and $b_2$.

As they are different, there will be one edge $e = (v,v')$ that occurs in $b_1$, but not in $b_2$.
As all weight disturbances are drawn independently, we assume without loss of generality that the weight disturbance to this edge is drawn last.

Now, the valuation $\val_2$ defined by $b_1$ does not depend on this final draw.
For $\val_2$, there is a value $w_e' = \val_2(v)-\lambda_e\val_2(v')$ that defines the weight $w_e'$ $e$ would need to have such that the inequation for $e$ is sharp.

For the valuation $\val_1$ defined by $b_1$ to be equal to $\val_2$, the weight for the edge $e$ (after adding the drawn distortion) needs to be exactly $w_e'$.
There is at most one value for the disturbance that would provide for this, and this disturbance for weight of $e$ is sampled with a likelihood of $0$.
\end{proof}

Putting these two results together, we get:

\begin{crl}
Given a pair of DPGs $\mathcal G = (V_{\min}, V_{\max}, E, w, \lambda)$ with contraction $\lambda^*$ and gap $\gamma$, and $\mathcal G' = (V_{\min}, V_{\max}, E, w', \lambda)$ obtained from $\mathcal G$ by adding independently uniformly at random drawn values from an interval $(-\varepsilon,\varepsilon)$ to every edge weight, for some $\varepsilon \leq \frac{1-\lambda^*}{3}\gamma$, then, a joint co-optimal strategy from $\mathcal G'$ is also co-optimal for $\mathcal G$, and $\mathcal G'$ is almost surely sharp.
\end{crl}

Note that we can estimate the gap cheaply when all coefficients in $\mathcal G$ are rational.
The gap is defined as the minimum of $\gamma_\sigma$ over non-co-optimal joint strategies $\sigma$, and we start by fixing such a joint strategy.

For a given $v\in V$, $\sigma$ defines a run $\rho = e_0 e_1 e_2 \ldots$ in the form of a "lasso path", which consists of a (possibly empty) initial path $e_0,\ldots,e_k$, followed by an infinitely often repeated cycle $e_0',\ldots,e_\ell'$, where the only vertex that occurs twice on the path $e_0,\ldots,e_k,e_0',\ldots,e_\ell'$ is the vertex reached before $e_0'$ and after $e_\ell'$, while all other vertices occur only once.

In this run, an edge $e_i$ occurs only once and contributes $\prod_{j=0}^{i-1}\lambda_{e_j} w_{e_i}$ to the value of this run, while an edge $e_i'$ occurs infinitely often and contributes the value $\frac{\prod_{j=0}^{k}\lambda_{e_j} \cdot \prod_{j=0}^{i-1}\lambda_{e_j'}}{1-\prod_{j=0}^{\ell}\lambda_{e_j'}}w_{e_i'}$ to the value of the run.

Now all we need to do is to find a common denominator.
To do this, let $\denom(r)$ be the denominator of a rational number $r$.
It is easy to see that 
$$\mathsf{common} = \prod\limits_{v\in V}\denom(\lambda_{(v,\sigma(v))})^2 \cdot \denom(w_{(v,\sigma(v))})$$
is a common denominator of the contributions of all of these weights, and thus a denominator that can be used for the sum.

Looking at an edge $e$ that defines $\gamma_\sigma$, then $\gamma_\sigma$ can be written with the denominator $$\mathsf{common}\cdot \denom(\lambda_ew_e)\ .$$
Obviously, the nominator is at least $1$.

Estimating $\gamma$ can thus be done by using the highest possible denominators available in $\mathcal G$. The representation of the resulting estimate $\gamma$ is polynomial in the size of $\mathcal G$.
\medskip

\noindent\textbf{Biased sum of offsets.}
We modify sharp games that are not improving. This can be done by redefining the function $\off$ as follows:
\[
  \off'(\val, (v,v')) =
\alpha_{(v,v')} \cdot \off(\val, (v,v'))
\]
Such $\off$ are defined for every edge $e = (v,v')$, where all $\alpha_e$ are positive numbers, which we call the \emph{offset factor}.
Based on this change, we re-define all offset definitions and the objective function with a primed version that uses these positive factors.

\begin{thm}
Let $\mathcal{G}= (V_{\min}, V_{\max}, E, w, \lambda)$ be an improving DPG for a given set of positive numbers $\{\alpha_e > 0 \mid e \in E\}$, $\sigma$ a strategy for both players, $\val$ an optimal solution returned at Line 5 of Algorithm \ref{alg:algname} for the adjusted function $f_\sigma'$, and let there be no local improvements of $\sigma$ for $\val$.
Then there is neighbouring valuation $\val''$ to $\val$ such that there is a better strategy $\sigma'$ that satisfies $f_{\sigma'}(\val'')<f_\sigma(\val)$.

Such a strategy $\sigma'$ is better than $\sigma$, and it can be selected in a way that $(v,\sigma'(v)) \in E_\val^\sigma$ holds for all $v\in V$ for a given set $E_\val^\sigma$.
\end{thm}

\begin{proof}
All proofs can be repeated verbatim with the new offset definition.
\end{proof}

\begin{thm}\label{t13}
If each offset factor $\alpha_e$ is selected independently uniformly at random from a bounded interval of positive numbers%
\footnote{or from any other distribution over positive numbers that has $0$ weight for all individual points}, 
then a sharp DPG $\mathcal{G}= (V_{\min}, V_{\max}, E, w, \lambda)$ is almost surely improving for a sampled set of positive numbers $\{\alpha_e > 0 \mid e \in E\}$.
\end{thm}
\begin{proof}
We show the claim for two arbitrary but fixed valuations $\val_1$ and $\val_2$ defined by two different bases, $b_1$ and $b_2$, respectively, that satisfy the inequations in $H$, and an arbitrary but fixed adjusted $f_\sigma$.
As there are finitely many bases and finitely many joint strategies, satisfying the requirement almost surely for them entails that the requirement is satisfied almost surely for the game.

As $\mathcal G$ is sharp, we have $\val_1 \neq \val_2$.

We first assume for contradiction that $\off(\val_1,(v,\sigma(v))) = \off(\val_2,(v,\sigma(v)))$ holds for all $v \in V$.
We pick a vertex $v$ such that $|\val_1(v) - \val_2(v)|>0$ is maximal, such a vertex exists since $\val_1 \neq \val_2$.
For this $v$, $\off(\val_1,(v,\sigma(v))) = \off(\val_2,(v,\sigma(v)))$ entails 
$\val_1(v) - \val_2(v) = \lambda_{(v,\sigma(v))} (\val_1(\sigma(v)) - \val_2(\sigma(v)))$.
Using $\lambda_e \in [0,1)$, we get
$|\val_1(\sigma(v)) - \val_2(\sigma(v))|>|\val_1(v) - \val_2(v)|>0$, which contradicts the maximality of $v$. Note that neither $\lambda_e$, nor the difference of $\val_1(\sigma(v))$ and $\val_2(\sigma(v))$ can be equal to 0 since $\val_1(v) \neq \val_2(v)$. Hence the contradiction.

We therefore have that $\off(\val_1,(v,\sigma(v))) \neq \off(\val_2,(v,\sigma(v)))$ holds for some $v\in V$.
As the $\alpha_e$ are drawn independently, we can assume w.l.o.g.\ that $\alpha_{(v,\sigma(v))}$ is drawn last.
There is at most one value $\alpha_{(v,\sigma(v))}'$ for which the condition
$$\sum_{v\in V} \off'(\val_1,(v,\sigma(v))) \neq \sum_{v\in V} \off'(\val_2,(v,\sigma(v)))$$
is not satisfied.

It therefore holds almost surely for all strategies that all base-induced valuations have a pairwise distinct image by the objective function associated to the strategy.
This immediately implies that the game is improving.
\end{proof}

Thus, we can almost surely obtain sharpness by adding small noise to the weights,
and almost surely make games improving by considering the offsets of the individual edges with a randomly chosen positive weight.
This guarantees cheap progress for the case where there are no local improvements.

\subsection{Mixing Pivoting on the Simplex and of the Objective}
\label{ssec:mixing}

When using a simplex based technique to implement $\mathsf{LinearProgramming}$ (Line 5 of Algorithm \ref{alg:algname}), then the algorithm mixes three
approaches that stepwise reduce the value of $f_\sigma(\val)$:
\begin{enumerate}
    \item The simplex algorithm updates the base, changing $\val$ (while retaining the objective function $f_\sigma$).
    \item Local updates, who change the objective function $f_\sigma$ (through updating $\sigma$)
    and retain $\val$.
    \item Non-local updates.
\end{enumerate}

Non-local updates are more complex than the other two, and the correctness proofs make use of the non-existence of the other two improvements. For both reasons, it seems natural to take non-local updates as a last resort.

The other two updates, however, can be freely mixed, as they both bring down the value of $f_\sigma(\val)$ by applying local changes.
That the improvements from (1) are given preference in the algorithm is a choice made to keep the implementation of the algorithm for using linear programs open, allowing, for example, to use ellipsoid methods \cite{Kha79} or inner point methods \cite{Kar84} to keep this step tractable.

\section{Discussion}
\label{sec:discuss}

There is widespread belief that mean payoff and discounted payoff games have two types of algorithmic solutions: value iteration~\cite{FGO20,Koz21} and strategy improvement~\cite{Lud95,Pur95,BV07,Sch08a,STV15}.
We have added a third method, which is structurally different and opens a new class of algorithms to attack these games.
Moreover, our new symmetric approach has the same proximity to linear programming as strategy improvement algorithms, which is an indicator of efficiency.

Naturally, a fresh approach opens the door to much follow-up research.
A first target for such research is the questions on how to arrange the selection of better strategies to obtain fewer updates, either proven on benchmarks, or theoretically in worst, average, or smoothed analysis.
In particular, it would be interesting to establish non-trivial upper or lower bounds for various pivoting rules. Without such a study, a trivial bound for the proposed approach is provided by the number of strategies (exponential). Moreover, the lack of a benchmarking framework for existing algorithms prevents us from testing and compare an eventual implementation.

A second question is whether this method as whole can be turned into an inner point method.
If so, this could be a first step on showing tractability of discounted payoff games -- which would immediately extend to mean-payoff and parity games.

\subsection*{Acknowledgements}

\includegraphics[height=8pt]{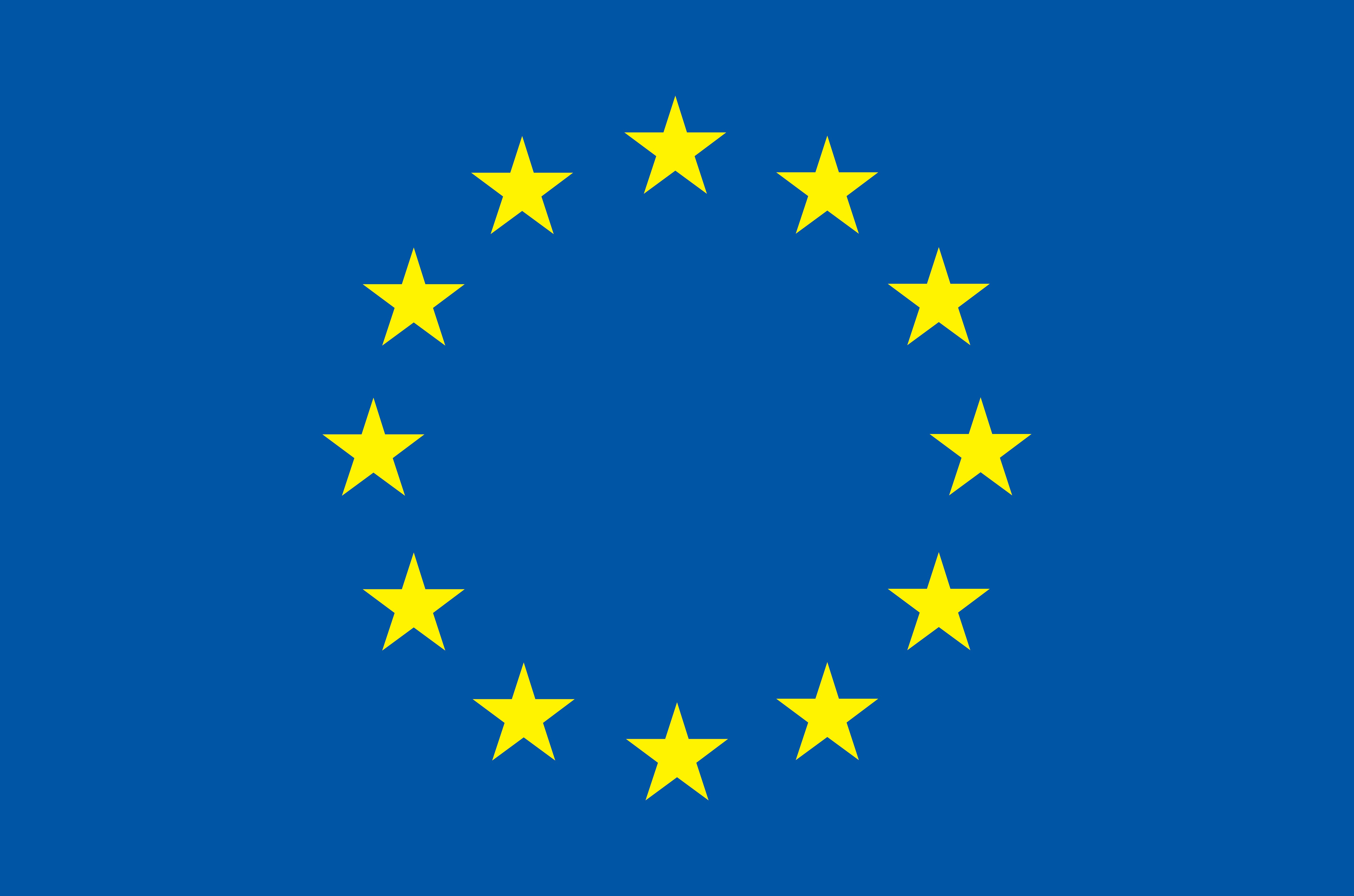} This project has received funding from the European Union’s Horizon 2020 research and innovation programme under the Marie Sk\l odowska-Curie grant agreement No 101032464.
It was supported by the EPSRC through the projects EP/X017796/1 (Below the Branches of Universal Trees) and
EP/X03688X/1 (TRUSTED: SecuriTy SummaRies for SecUre SofTwarE Development).

\bibliographystyle{eptcs}
\bibliography{bib}

\end{document}